\documentclass[review]{elsarticle}
\usepackage{microtype}
\usepackage[fleqn]{amsmath}
\usepackage{graphics}
\usepackage{amsthm}
\usepackage{graphicx, amssymb, color}
\usepackage[margin=2.5cm]{geometry}
\usepackage{mathtools}
\usepackage{float}
\graphicspath{{./}}

\usepackage{algorithmic}			
\usepackage{algorithm}
\usepackage{placeins}
\usepackage{hyperref}

\newtheorem{algTitle}{Algorithm}
\newtheorem{definition}{Definition}

\newtheorem{proposition}{Proposition}
\newtheorem{corollary}{Corollary}
\newtheorem{remark}{Remark}

\newcommand*{\defeq}{\stackrel{\text{def}}{=}}

\newcommand{\argmin}{\operatornamewithlimits{argmin}}

\begin{document}
\begin{frontmatter}
\title{Computer Aided Restoration of Handwritten Character Strokes}

\author{Barak Sober}
\author{David Levin}
\address{School of Mathematical Sciences, Tel-Aviv University, 55 Haim Levanon St., Tel Aviv 6997801, Israel}

\maketitle
\begin{abstract}
	This work suggests a new variational approach to the task of computer aided segmentation and restoration of incomplete characters, residing in a highly noisy document image. We model character strokes as the movement of a pen with a varying radius. Following this model, in order to fit the digital image, a cubic spline representation is being utilized to perform gradient descent steps, while maintaining interpolation at some initial (manually sampled) points. The proposed algorithm was used in the process of restoring approximately 1000 ancient Hebrew characters (dating to ca. $8^{th}$-$7^{th}$ century BCE), some of which are presented herein and show that the algorithm yields plausible results when applied on deteriorated documents. 
\end{abstract}

\end{frontmatter}

\section{Introduction}
The initial impetus for this work comes from the field of First Temple Period (i.e., Iron Age in Israel) Hebrew epigraphy (the study of the development of writing through time). Most of the surviving inscriptions dating to that period in Israel are written with ink  on potshards (ostraca). These inscriptions are very badly preserved, as they spent three millenia underground, and the surviving characters are incomplete, noisy and sometimes barely legible (see Figure \ref{fig:Ostraca}). 

\begin{figure}[ht]
	\begin{centering}
		\includegraphics[width={0.8\linewidth}]{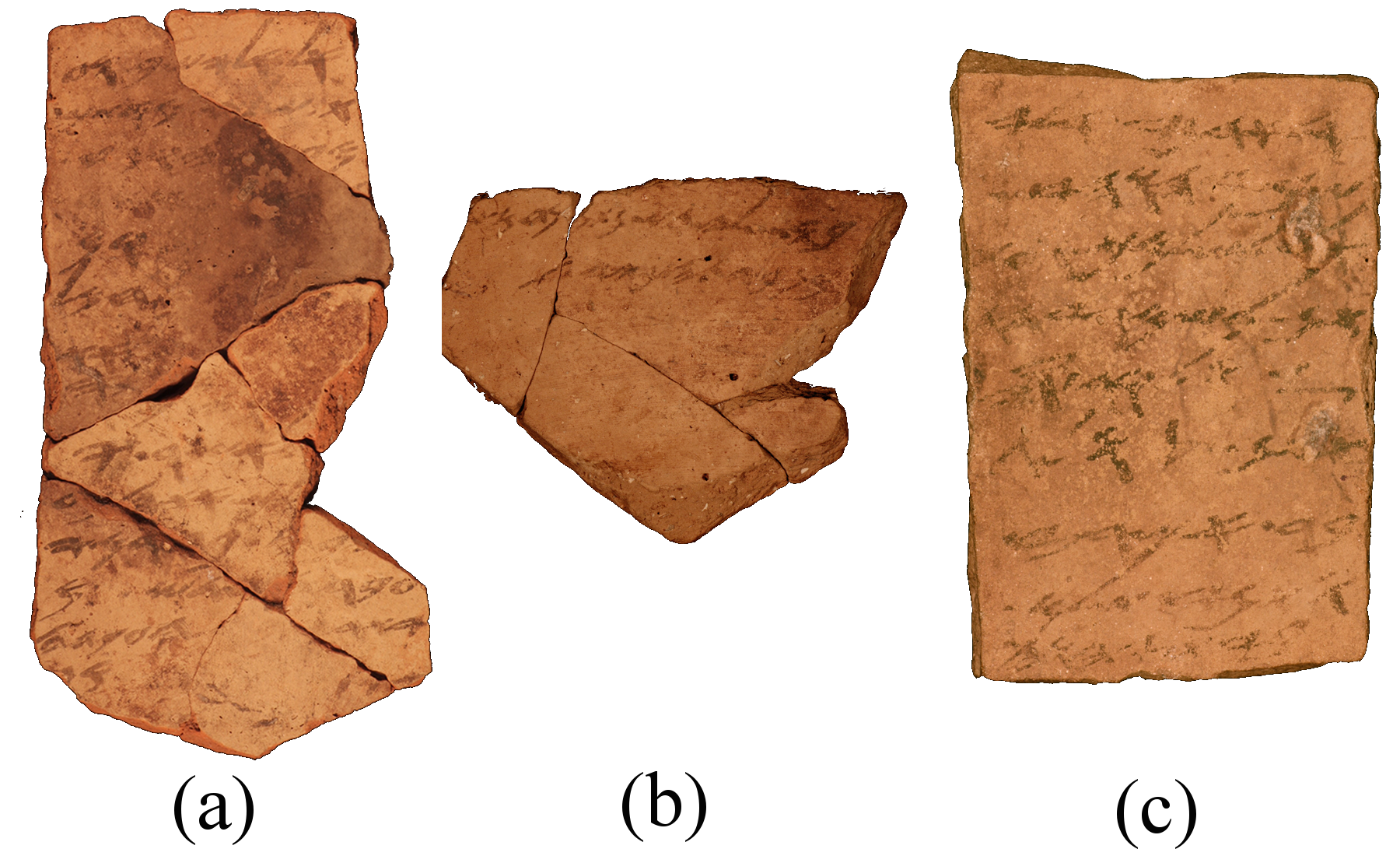}
		\par\end{centering}
		\caption{Ostraca images: (a) ostracon No. 5 from Tel Arad; (b) ostracon No. 17 from Tel Arad; and (c) ostracon No. 18 from Tel Arad. Even the better preserved ostraca are blurred and partly effaced. \label{fig:Ostraca}}
\end{figure}

Traditionally, a newly discovered inscription is first being transcribed manually to create a handmade binarization of the document (facsimile - Fig. \ref{fig:OstFacs}). These facsimiles are then used to decipher the meaning of the transcribed document, as well as to date inscriptions from unclear contexts \cite{rollston1999script}.
  However, the process of creating such transcriptions is subject to human error and in many cases mixes documentation with interpretation. For example, Figure \ref{fig:ostrafacs_problem} shows an instance of a transcribed character that did not exist in the original text. As a consequence, one cannot rely on these human-made binarizations as a basis for other document analysis tasks (be it automatic or manual).  

\begin{figure}[ht]
	\begin{centering}
		\includegraphics[width={\linewidth}]{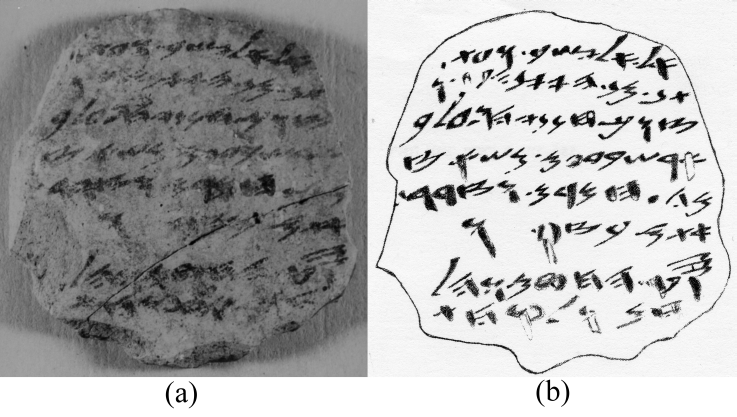}
		\par\end{centering}
	\caption{Ostracon No. 3 from Tel Arad: (a) grayscale image; (b) handmade facsimile of the ostracon. The facsimile was originally published in \cite{AharoniArad}. \label{fig:OstFacs}}
\end{figure}

In numerous document analysis tasks, the initial procedure is the creation of a binary image, separating the foreground (characters) from the background. However, in cases of severely damaged documents, such as the ones we are dealing with, the existing algorithms fail to give a satisfactory result (see \cite{shaus2012bin} for example). The reason for this lies in the following facts: (a) the original writing is occasionally erased and therefore its binarization would justifiably be incomplete; (b) the imperfections induced by the long deterioration process might seem as an authentic signal, hence the result would contain unwanted segments (Fig. \ref{fig:OstBin}). Moreover, in contrary to common document analysis systems, where the binarization is just a pre-processing step in a pipeline, the characters' restorations themselves are of great importance to the field of Iron Age Hebrew epigraphy.

\begin{figure}[ht]
	\begin{centering}
		\includegraphics[width={\linewidth}]{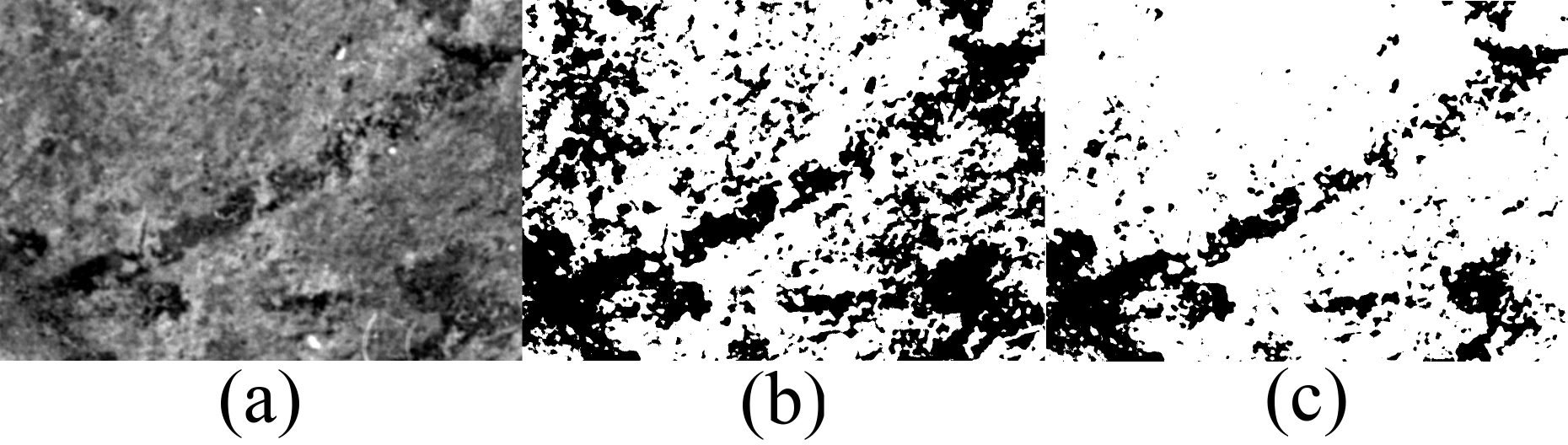}
		\par\end{centering}
	\caption{Example of a character taken from Arad ostracon No. 18: (a) grayscale image; (b) binarization using Otsu's method \cite{Otsu}; (c) binarization using Sauvola's method \cite{Sauvola}.\label{fig:OstBin}}
\end{figure}

Therefore, we have undertaken to tackle this issue by introducing mathematical techniques from the fields of computer aided design (CAD) and image processing in order to create a semi-automatic approximation of the original characters. Naturally, the resulting restorations will have common characteristics, which will simplify possible future tasks of comparing between scripts originating from different periods or from different regions (i.e., paleography).

\begin{figure}
	\begin{centering}
		\includegraphics[width={\linewidth}]{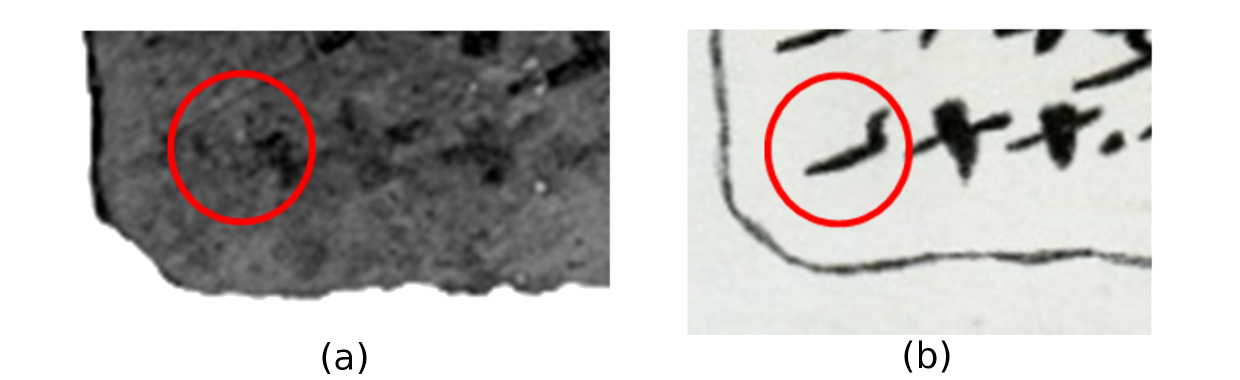}
		\par\end{centering}
	\caption{Ostracon No. 1 from Arad (a) original image; (b) facsimile of the zoomed region (painted by Y. Aharoni \cite{AharoniArad} courtesy of Israel Exploration Society). The marked letter cannot be recognized in the original image, and was completed according to the context of the word.\label{fig:ostrafacs_problem}}
\end{figure}

There have been other attempts to introduce mathematical tools to the field of epigraphy but this discipline of research is still in its infant steps (recent work on digital epigraphy can be found in \cite{wolf2011computerized, fischer2010codicology, diem2011recognizing, lavrenko2004holistic, Papa2014WriterIdentification2, Papa2009WriterIdentification1}, some works related to Iron Age Hebrew epigraphy can be found in \cite{ms2012joas, ms2014peq, msMalhata, shaus2012bin, shaus2012quality, GlyphEvaluation}). Nevertheless, those attempts are not connected directly to the techniques presented below. There are works dealing with reconstruction of damaged handwritten characters as a result of graphics removal (e.g., see the works \cite{lopresti2010ruling, abd2009page, arvind2007line, cao2002text}). However, the deficiencies dealt with in these articles are relatively easy to model, in contrary to the case of natural deterioration processes. Other research fields that are more closely related to our stroke restoration algorithm, are Approximation using splines \cite{de1978practical}; Active Contours \cite{kass1988snakes} and Variational methods in image processing \cite{mumford1989optimal, SSVM2011}.  

\section{Stroke Restoration}

\subsection{Modeling a stroke}

Many models have been proposed in the last decades to study human movement in general and handwritten strokes in particular. The various classical models are based on neural networks (e.g., \cite{schomaker1991simulation, gangadhar2007oscillatory}), equilibrium point models (e.g., \cite{feldman1966functional,feldman2005testing}), behavioral models (e.g., \cite{van1983independent}), kinematic models (e.g., \cite{plamondon1995kinematic, plamondon2007extraction}), and models relying on minimization principals (e.g., \cite{flash1985coordination, edelman1987model, wada1995theory}). For a comprehensive survey regarding such models see the introduction to the article \cite{plamondon2014recent}.  
These models aim at deciphering the way the underlying human cognitive system, generating the movement, works. Instead, we aim at approximating an already existing incomplete stroke. Therefore, it is sufficient to use a far simpler representation of the stroke, disregarding the sophisticated mechanism used to create it (see \cite{ICDAR2005SplineStroke} for such a stroke definition - closely related to the one we will use here). 

Keeping the simplicity in mind, a stroke can be referred to as a two-dimensional piecewise smooth curve in some parameter $t \in [a,b] $. However, such a representation ignores the stroke's thickness dimension,
which is related to the stance of the writing pen towards the document (in our case - potshard) and to the characteristics of the pen itself. In the case of Iron Age Hebrew it is well accepted that the scribes used reed pens, which have a flat top rather than pointed. This fact makes the writing thickness even more essential to the process of stroke restoration. 

In accordance to what have been stated above, the restoration of a stroke would be a mere approximation of a sampled thick curve. From here on, we shall use the term stroke in the following sense:  
\begin{definition}
	A \textit{stroke} is a piecewise-smooth part of a character with a specific
	radius of writing at each point, resulting from the act of writing.
	A stroke starts when the pen touches the surface of writing and ends
	when the pen is lifted. 
	\label{def:stroke}
\end{definition}
In mathematical terms we denote the \textit{stroke} as a set-valued function:
 
\begin{equation}
S(t)\defeq \{(p,q)\mid(p-x(t))^{2}+(q-y(t))^{2}\leq r(t)^{2}\} ~,~t \in [a,b]
\end{equation}
 
where $x(t)$ and $y(t)$ represent the coordinates of the center
of the pen at time $t$, and $r(t)$ stands for the pen's radius at $t$
(see Figure \ref{fig:Letter-E}). The corresponding \textit{stroke curve} is thus:
\[
\gamma(t)=(x(t),y(t),r(t))^T \qquad t\in[a,b]
,\]
 
while the \textit{skeleton of the stroke} will accordingly be the curve
 
\[\beta(t)=(x(t),y(t))^T \qquad t\in[a,b].\]

\begin{figure}
	\begin{centering}
		\includegraphics[width={0.6\linewidth}]{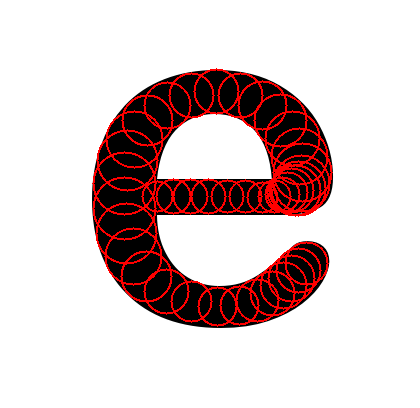}
		\par\end{centering}
	
	\caption{The character "e" comprised of discs. The discs painted in red over the character were created using the \textit{stroke restoration} algorithm (described in Algorithm \eqref{alg:StrokeRestoration}).\label{fig:Letter-E}}
	
\end{figure}

\subsection{Energy Minimization}

The use of energy functional minimization frameworks is widespread in the image processing literature. Two such canonical examples are the active contours \cite{kass1988snakes} and the Mumford-Shah framework \cite{mumford1989optimal}, both referring to the problem of segmentation. Borrowing the idea of minimizing an energy functional, we produce an analytic reconstruction of a stroke with respect to a given image $I(p,q) \quad \left( (p,q) \in [1,N]\times[1,M] \right)$. This reconstructed stroke $S^{*}(t)$, corresponding to the stroke-curve $\gamma^{*}(t)=(x^{*}(t),y^{*}(t),r^{*}(t))^T$, is the one that minimizes the following functional:
 
\begin{equation}
		F[\gamma(t)]=
		c_{1}\intop_{0}^{T}\frac{G_{I}(t)}{(r(t))^{2}}dt+c_{2}\intop_{0}^{T}\frac{1}{\sqrt{r(t)}}dt 
		+c_{3} \sum\limits_{k=0}^{N-1} \intop_{t_k + \epsilon}^{t_{k+1} - \epsilon}|K(\dot{x},\dot{y},\ddot{x},\ddot{y})|dt
	\label{eq:final_functional}
	.\end{equation}

\begin{equation}
	\gamma^{*}(t)=\argmin_{\gamma(t)}\{F[\gamma(t)]\}
	\label{eq:argmin_func}
	,\end{equation}
 
where $G_{I}(t)=\sum_{(p,q)\in S(t)} I(p,q)$ is the summation of the gray level values of the image $I$ inside the disc $S(t)$; $\dot{x}, \ddot{x}$ and $\dot{y}, \ddot{y}$ denote the first and second derivatives of $x$ and $y$ with respect to the parameter $t$; $K(\dot{x},\dot{y},\ddot{x},\ddot{y})=\frac{\dot{x}\ddot{y}-\dot{y}\ddot{x}}{(\dot{x}^{2}+\dot{y}^{2})^{3/2}}$ stands for the curvature of the skeleton of the stroke $\beta(t)$; $0<c_{1},c_{2},c_{3}\in\mathbb{R}$ are balancing parameters; and $0<\epsilon \in\mathbb{R}$ is of a small magnitude. 

The reconstruction is subject to initial and boundary conditions at $t_k$ manually sampled when one of the following terms hold: a) beginning and end of strokes; b) intersections of strokes; c) significant extremal points of the curvature; d) points with no traces of ink. For more details regarding the manual sampling see Section \eqref{sec:TheAlgorithm}.

\textit{First and Second Terms \textendash{} Fidelity:} We denote the right hand side elements of equation \eqref{eq:final_functional} by:
 
\[
F_{E}^{(S)} = c_{1}\intop_{0}^{T}\frac{G_{I}(t)}{(r(t))^{2}}dt
,\]
 
 and
 
\[
F_{E}^{(r)} = c_{2}\intop_{0}^{T}\frac{1}{\sqrt{r(t)}}dt
.\]
 
The motivation behind these terms is to force the minimum of the functional to correspond to the original image of the stroke. $F_{E}^{(S)}$ is an integration of the gray level values over the reconstructed stroke which should be low if the reconstructed stroke covers the image of the stroke, while $F_{E}^{(r)}$ is lower when the radii of the discs become larger. Thus preventing the solution from collapsing to the null set at each $t$. The balance between these two terms is of great importance and therefore a detailed discussion regarding this issue is presented in Appendix A.

\textit{Third Term \textendash{} Internal Energy:} As opposed to the first and second terms, the third term on the right hand side of equation \eqref{eq:final_functional} 
 
\[
F_K = c_{3} \sum\limits_{k=0}^{N-1} \intop_{t_k + \epsilon}^{t_{k+1} - \epsilon}|K(\dot{x},\dot{y},\ddot{x},\ddot{y})|dt
,\]
 
does not relate to the image of the handwriting, but to the approximation of the stroke's skeleton alone. This term would have been redundant had we an ideal image of the stroke, as the restoration would be a close fit to the image. However, in our case where the image is severely damaged and the stroke is often incomplete, it is vital to control the smoothness of the curve. Accordingly, this term limits the possibility of high curvature areas and hence keeps the curve smoother.

We stated above that a stroke starts where the pen touches the surface and ends where the pen is lifted (see Definition \ref{def:stroke}). Therefore, using the integral over the curvature of the stroke's skeleton could prove to be problematic, as there might exist local curvature maxima and corners (i.e., discontinuity of the second derivative). In order to avoid the "smoothing" of desired corners the skeleton's curvature is ignored in $\epsilon$-neighborhoods of the sampled points.

\section{Algorithms}
Prior to presenting our algorithm, we begin by narrowing down the domain of possible solutions for the minimization problem presented above. Next, following an analysis of the reduced problem, we provide a detailed description of the \textit{stroke restoration} algorithm. We conclude this section with some remarks regarding possible improvements.

\subsection{Limiting The Solutions Domain}
In order to simplify the minimization problem at hand, we restrict the space of possible reconstructing curves $\gamma (t) = \left( x(t), y(t), r(t) \right)^T$ to \textbf{natural cubic splines} (with respect to some nodes $\lbrace t_j \rbrace_{j=-1}^{n+1}$). It is noteworthy that this restriction is adequate since such curves minimize the integral over the squared magnitude of the second derivative, and therefore controls the curvature as well (e.g., see \cite{SplineMinimalSecondDerivative}). This quality fits the third term in equation \eqref{eq:final_functional}. Accordingly, we can now represent such a curve in the cubic B-Spline basis for equally distributed nodes:
 
\begin{equation}
	\gamma (t) 
	= \sum\limits_{j=-1}^{n+1} \left( c_j ^{x}, c_j ^{y} , c_j ^{r} \right)^{T} B_j (t) 
	= \sum\limits_{j=-1}^{n+1} \vec{c} B_j (t)
	\label{eq:spline_rep}
	,\end{equation}
 
Using this representation, each curve is now described in full by $n+3$ coefficients vectors $\vec{c_i}$ which will be referred to henceforth as \textit{control points}. Thus, we can rewrite equation \eqref{eq:final_functional} as a function of $d = 3 \cdot (n+3)$ variables:
 
\begin{equation}
	\begin{split}
		F[\gamma]
		&= F \left( c_{-1}^{x} , c_{-1}^{y} , c_{-1}^{r} , ... , c_{n+1}^{x} , c_{n+1}^{y} , c_{n+1}^{r} \right) \\
		&= F \left( y_1 , ... , y_d\right)
	\end{split}
	\label{eq:functional_for_grad}
	,\end{equation}
 
where $\gamma (t)$ is defined by \eqref{eq:spline_rep}, and $y_1 , ... , y_d$ are the variables of $F$ corresponding to the relevant indices. Accordingly, we have:
 
\[
\left(  
\begin{array}{c}
y_{3(i + 2) - 2} \\
y_{3(i + 2) - 1} \\
y_{3(i + 2)}
\end{array}  
\right) = 
\left(  
\begin{array}{c}
c_{i}^{x} \\
c_{i}^{y} \\
c_{i}^{r}
\end{array}  
\right) = 
\vec{c}_{i}
\]

Following this rationale, given an initial guess we can apply the Gradient Descent method in order to search for the optimal solution. The convergence of the algorithm is not guaranteed theoretically, as we do not know whether the problem is convex in the general setting. Nevertheless, our experiments with a slightly modified Gradient Descent (described in the following passages) showed that given a piecewise linear first guess (using just a few manually sampled points), the algorithm yielded reliable restorations. For more details please see Sections \ref{sec:TheAlgorithm} and \ref{sec:Experimental}.

The classical Gradient Descent algorithm for unconstrained optimization of $F:\mathbb{R}^d \rightarrow \mathbb{R}$ is defined by the following iterative process \cite{barzilai1988two}:
 \begin{equation}
	\label{eq:GradDescent}
	x_{k+1} = x_k - \alpha_k \nabla F(x_k)
	,\end{equation} 
and $\alpha_k \in \mathbb{R}$ is given by
 \[
\alpha_k = \argmin_{\alpha} \lbrace F \left( x_k - \alpha \nabla F(x_k) \right) \rbrace
.\] 

Unfortunately, applying this method to our case proved to be inefficient as it converges to insignificant local minima stemming from the noise of the given data. As a consequence, we used a
different step-size for each element in our domain so that equation \eqref{eq:GradDescent} now becomes:
 \begin{equation}
	x_{k+1} = x_k - \left( \alpha_1^{(k)} \frac{\partial F}{\partial y_1}(x_k) , ... , \alpha_d^{(k)} \frac{\partial F}{\partial y_d}(x_k) \right)
	\label{eq:GradDescentAltered}
	,\end{equation} 
We start by assigning some initial values to $\alpha_i^{(k)}$ and update the values at each iteration according to:
 \[
\alpha_i^{(k+1)} = 
\begin{cases}
\alpha_i^{(k)} \ & \frac{\partial F_{y_i}}{\partial y_i}(x_k) \cdot \frac{\partial F}{\partial y_i}(x_{k-1})  >0  \\ 
\alpha_i^{(k)} \cdot T \ & otherwise
\end{cases}
,\] 
where $T \in (0,1)$. This way, each time the sign of the directional derivative changes, $\alpha_i^{(k)}$ decreases so that the steps would be smaller in that direction. This method has proven empirically to be significantly more stable.

\subsection{Performing Gradient Descent and maintaining interpolation} 
As we wish to present a semi-automatic procedure, there is a need to maintain the interpolation at the initial and boundary points, sampled by the user. For this end, the Gradient Descent steps should not interfere with the interpolation conditions. In what follows, we develop a way to perform Gradient Descent steps with respect to the spline's \textit{control points}, while maintaining the interpolation at the desired points.

Applying the standard Gradient Descent steps described in equation \eqref{eq:GradDescent} to our case (i.e., using the representation from \eqref{eq:functional_for_grad}) results in the following iterative process:
 \begin{equation}
	\label{eq:InitStrokeIter}
	\vec{C}^{(k+1)} = \vec{C}^{(k)} - \alpha^{(k)} \nabla F(\vec{C}^{(k)})
	,\end{equation} 
where we denote 
 \[ \vec{C}^{(k)} =  
\left( \begin{array}{ccc}
| &  & | \\
\vec{c}_{-1}^{(k)} & \cdots & \vec{c}_{n+1}^{(k)} \\
| &  & | 
\end{array} \right)
, \] 
$\vec{c}_{-1}^{(k)} , ... , \vec{c}_{n+1}^{(k)}$ are the control points at the $k^{th}$ iteration (i.e., $ \gamma ^{(k)} (t) = \sum\limits_{j=-1}^{n+1}\vec{c}_j^{(k)} B_j(t) $ ), and
 \[
\nabla F(\vec{C}^{(k)}) = 
\left( \begin{array}{ccc}
\frac{\partial F}{\partial c_{-1}^{x}}(\vec{C}^{(k)}) & \cdots & \frac{\partial F}{\partial c_{n+1}^{x}}(\vec{C}^{(k)}) \\
\frac{\partial F}{\partial c_{-1}^{y}}(\vec{C}^{(k)}) & \cdots & \frac{\partial F}{\partial c_{n+1}^{y}}(\vec{C}^{(k)})  \\
\frac{\partial F}{\partial c_{-1}^{r}}(\vec{C}^{(k)}) & \cdots & \frac{\partial F}{\partial c_{n+1}^{r}}(\vec{C}^{(k)}) \end{array} \right)
.\] 

It is easy to verify that in order to satisfy $ \gamma ^{(k)} (t_i) = \vec{f}_i $ for some node at $t_i$ at the $k^{th}$ iteration, the control points must maintain the following relation:
 \begin{equation}
	\frac{1}{6} \vec{c}_{i-1}^{(k)} + \frac{2}{3} \vec{c}_{i}^{(k)} + \frac{1}{6} \vec{c}_{i+1}^{(k)}  = \vec{f}_i
	.\end{equation} 
Hence, if we maintain this condition at each iteration we have
\small
 \[
\begin{split}
&\frac{1}{6} ( \vec{c}_{i-1}^{(k)} - \vec{c}_{i-1}^{(k+1)} ) + 
\frac{2}{3} ( \vec{c}_{i}^{(k)} - \vec{c}_{i}^{(k+1)} ) \\
&+ \frac{1}{6} (\vec{c}_{i+1}^{(k)} - \vec{c}_{i+1}^{(k+1)}) 
= \vec{0}
\end{split}.\] 
\normalsize
So, by denoting $\vec{\delta}_i = \vec{c}_i^{(k)} - \vec{c}_i^{(k+1)}$ we get:
 \[
\frac{1}{6} \vec{\delta}_{i-1} + 
\frac{2}{3} \vec{\delta}_{i} + 
\frac{1}{6} \vec{\delta}_{i+1}  = \vec{0}
\] 
 \begin{equation}
	\label{eq:deltaRelation}
	\vec{\delta}_{i-1} + 
	4 \vec{\delta}_{i} + 
	\vec{\delta}_{i+1}  = \vec{0}
\end{equation} 

However, by rewriting equation \eqref{eq:InitStrokeIter} we get the following relation:
 \[
\vec{C}^{(k+1)} - \vec{C}^{(k)}=  - \alpha^{(k)} \nabla F(\vec{C}^{(k)})
\] 
 \[
\left( \begin{array}{ccc}
| & & | \\
\vec{c}_{-1}^{(k+1)} - \vec{c}_{-1}^{(k)} & \cdots & \vec{c}_{n+1}^{(k+1)} - \vec{c}_{n+1}^{(k)} \\
| & & |
\end{array} \right) = 
-\alpha^{(k)} \left( \begin{array}{ccc}
\frac{\partial F}{\partial c_{-1}^{x}}(\vec{C}^{(k)}) & \cdots & \frac{\partial F}{\partial c_{n+1}^{x}}(\vec{C}^{(k)}) \\
\frac{\partial F}{\partial c_{-1}^{y}}(\vec{C}^{(k)}) & \cdots & \frac{\partial F}{\partial c_{n+1}^{y}}(\vec{C}^{(k)})  \\
\frac{\partial F}{\partial c_{-1}^{r}}(\vec{C}^{(k)}) & \cdots & \frac{\partial F}{\partial c_{n+1}^{r}}(\vec{C}^{(k)}) \end{array} \right)
.\]  
Hence, we can rewrite $\vec{\delta}_i$ as:
 \[
\vec{\delta}_i = 
\alpha^{(k)} \left( \begin{array}{c}
\frac{\partial F}{\partial c_{i}^{x}}(\vec{C}^{(k)}) \\
\frac{\partial F}{\partial c_{i}^{y}}(\vec{C}^{(k)}) \\
\frac{\partial F}{\partial c_{i}^{r}}(\vec{C}^{(k)})
\end{array} \right)
.\] 
Accordingly, we can reinterpret equation \eqref{eq:deltaRelation} as a condition over the directional derivatives of our function $F$. Explicitly, we get the following condition for $x , y$ and $r$ independently:
 \[
\frac{\partial F}{\partial c_{i-1}}+ 
4 \frac{\partial F}{\partial c_{i}} + 
\frac{\partial F}{\partial c_{i+1}}  = 0
\] 

This relation restricts our problem and reduces the dimensionality of the derivation directions in our gradient. Thus, we should not find the gradient with respect to the standard directions but use directions that do not interfere with the interpolation conditions. 

If we represent the standard directions as direction vectors we get the standard derivation basis:
 \[
\left( \begin{array}{cccc}
1 & 0 & \cdots & 0 \\
0 & 1 & \ddots & \vdots \\
\vdots & \ddots & \ddots & 0 \\
0 & \cdots & 0 & 1 
\end{array} \right)
.\] 
Thus, in order to uphold the interpolatory conditions in the $i^{th}$ node, instead of using the three standard directions:
\small
 \[
\begin{array}{cccc}
(i) &
\left[ \begin{array}{c} 
0 \\ \vdots \\ 1 \\ 0 \\ 0 \\ \vdots \\ 0
\end{array} \right] &
\left[ \begin{array}{c} 
0 \\ \vdots \\ 0 \\ 1 \\ 0 \\ \vdots \\ 0
\end{array} \right] &
\left[ \begin{array}{c} 
0 \\ \vdots \\ 0 \\ 0 \\ 1 \\ \vdots \\ 0
\end{array} \right]
\end{array}
,\] 
\normalsize
we would use the following directions:
\small
 \begin{equation}
	\label{eq:derivationBasis}
	\begin{array}{ccc}
		(i) &
		\left[ \begin{array}{c} 
			0 \\ \vdots \\ 1 \\ - \frac{1}{4} \\ 0 \\ \vdots \\ 0
		\end{array} \right] &
		\left[ \begin{array}{c} 
			0 \\ \vdots \\ 0 \\ - \frac{1}{4} \\ 1 \\ \vdots \\ 0
		\end{array} \right] 
	\end{array}
	.\end{equation} 
\normalsize
Moving along these directions will ensure the interpolation criterion is kept. This way for each constraint we lose a dimension in the derivation basis and change the directions accordingly. 

\begin{proposition}
	\label{prop:interpolatingStep}
Let $\gamma(t) = \sum\limits_{j=-1}^{n+1} c_j B_j(t)$ be a spline curve with respect to the nodes $t_j$ (where $B_j(t)$ are the B-spline basis functions), and assume $\gamma(t_i) = f_i$ at some node $t_i$. Then changing the coefficients $c_{i-1} , c_i , c_{i+1}$ along the directions presented in \eqref{eq:derivationBasis} does not affect the interpolation at $t_i$.
\end{proposition}
\begin{proof}
We wish to verify whether the interpolation is kept by $\tilde{\gamma}(t) = \sum\limits_{j=-1}^{n+1} \tilde{c}_j B_j(t)$, where $\tilde{c}_j = c_j$ for $j \neq i-1, i , i+1$ and 
	 \begin{equation}
		\label{eq:stepRelation}
		\begin{array}{l}
			\tilde{c}_{i-1} = c_{i-1} + \alpha \\
			\tilde{c}_{i} = c_{i} - \frac{1}{4} \alpha  - \frac{1}{4} \beta\\
			\tilde{c}_{i+1} = c_{i+1} + \beta
		\end{array}
		.\end{equation} 
	In order to do so it is sufficient to verify that:
	 \[
	\frac{1}{6} \tilde{c}_{i-1} + \frac{2}{3} \tilde{c}_{i} + \frac{1}{6} \tilde{c}_{i+1} = \frac{1}{6} c_{i-1} + \frac{2}{3} c_{i} + \frac{1}{6} c_{i+1}
	,\] 
	and this is easily verified from equation \eqref{eq:stepRelation}.
\end{proof}

\begin{corollary}
	Using both the original and the slightly altered Gradient Descent steps described in equations \eqref{eq:GradDescent} and \eqref{eq:GradDescentAltered} (respectively) in the directions described in equation \eqref{eq:derivationBasis} uphold the interpolation conditions at the node $t_i$.
\end{corollary}

\begin{remark}
	In order to be able to perform the above mentioned step, there must exist at least one non-interpolating node in between two nodes where we wish to satisfy the interpolation conditions.
\end{remark}

\begin{remark}
	The directional derivative at each point is approximated by the forward difference scheme:
	 \[
	\frac{\partial F}{\partial v} \approx \frac{F(x + hv) - f(x)}{h \lvert v\rvert}
	.\] 
	Therefore, we first normalize our derivation basis to get 
	 \[
	\frac{\partial F}{\partial v} \approx \frac{F(x + hv) - f(x)}{h}
	.\] 
\end{remark}

\subsection{Stroke Restoration}
\label{sec:TheAlgorithm}

\begin{algTitle}
	\textit{Stroke Restoration}
	\label{alg:StrokeRestoration}
	\begin{enumerate}
		\item A user manually selects the initial and boundary points for the desired stroke (see Figure \ref{fig:AlgorithmSteps}a)
		\item An initial piecewise linear spline, interpolating the sampled points, is constructed (see Figure \ref{fig:AlgorithmSteps}b). The interpolation points are regarded as node points. Between each pair of interpolated points another node is created (For example: if a user sampled three points we will have a total of 5 nodes). We define the curve parameter such that the nodes $t_i$ are in  $ \mathbb{Z} $. This way the nodes are equally distributed.
		\begin{enumerate}
			\item For example, the first manually sampled point will correspond to the node at $t_0=0$, the second point will be connected with the node at $t_2=2$ and in between them another node would be created at $t_1=1$
		\end{enumerate}
		\item The spline's control points are developed, using the Gradient Descent iterations described at \eqref{eq:GradDescentAltered} with respect to the cost functional described in \eqref{eq:functional_for_grad} while maintaining the interpolation at the sampled points.
		\begin{enumerate}
			\item We set minimal and maximal radius parameters and do not allow the coefficients $c_i^{r}$ to exceed these limits.
		\end{enumerate}
	\end{enumerate}
\end{algTitle}

\begin{figure}[ht]
	\begin{centering}
		\includegraphics[width={\linewidth}]{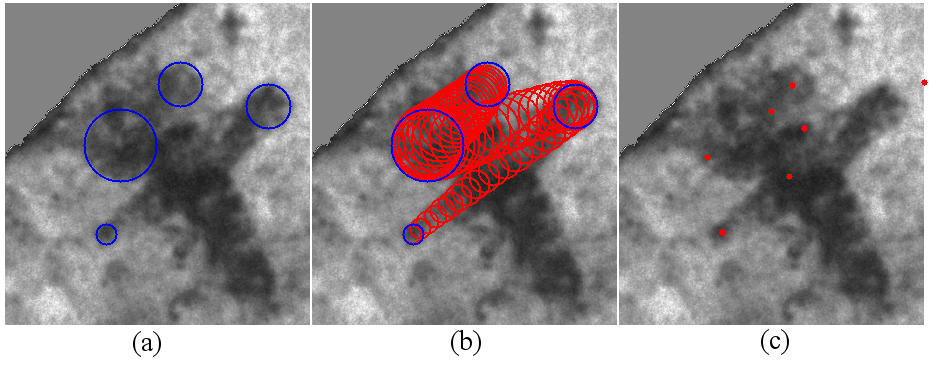}
		\par\end{centering}
	\caption{A character from Arad Ostracon No.1 (a) initial and boundary points sampled manually (b) the initial piecewise linear interpolation (c) the control points of the initial spline.}
	\label{fig:AlgorithmSteps}
\end{figure}

The manual selection of points is done in accordance to the following criteria:
\begin{enumerate}
	\item Choose the beginning and ending of strokes.
	\item Whenever there exist an intersection between two strokes sample it.
	\item Significant extremal points of the curvature should be sampled.
	\item It is preferred to sample areas where the ink is missing (i.e., points of discontinuity).
	\item If a stroke segment between two consecutive points is long then we should sample another one between them to allow the spline to represent the stroke as close to reality as possible.
\end{enumerate}
It is noteworthy that this process inserts a certain amount of subjectivity into the restoration, and can even be laborious, if a large number of characters are being worked. However, this obstacle can be overcome in the future, by automating the selection process (e.g., by using prior knowledge about the prototype of the character we are aiming to reconstruct).  

In Figure \ref{fig:StepbyStepIteration_waw} we can see the gradient descent steps described above applied to a "waw" character from Arad Ostracon No. 1. The initial and boundary points used in this run (marked in blue) and the initial piecewise linear spline are shown in Figure \ref{fig:StepbyStepIteration_waw}a. The algorithm  started from four manually sampled points and developed from a coarse approximation to a fine representation of the stroke.

\begin{figure}[ht]
	\begin{centering}
		\includegraphics[width={\linewidth}]{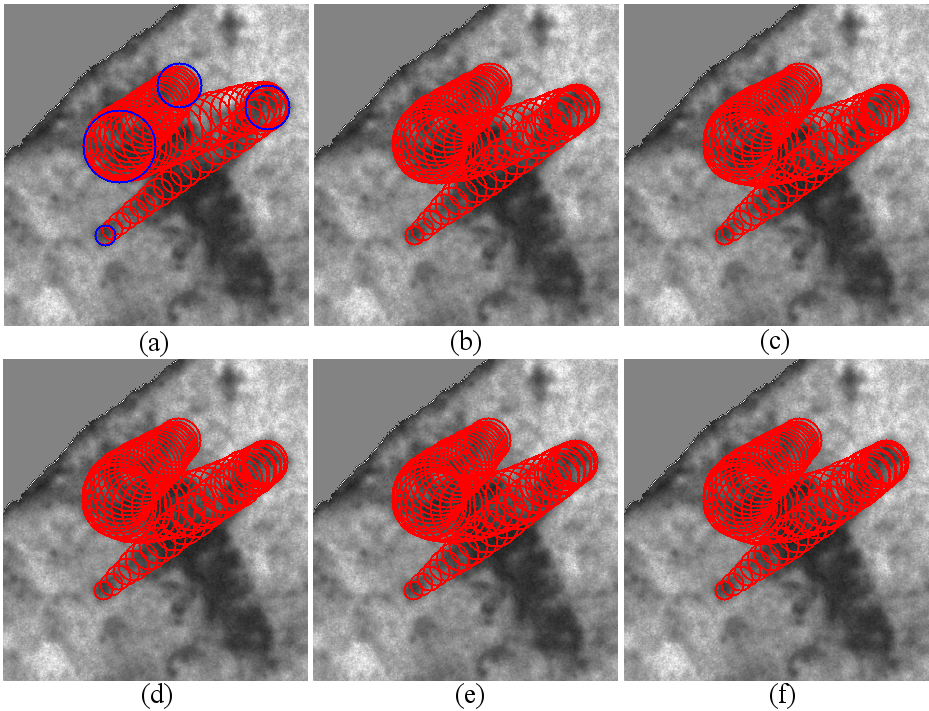}
		\par\end{centering}
	\caption{Restoration of a stroke from a "waw" character taken from Arad Ostracon No.1. (a) The initial piecewise-linear spline. Marked in blue are the manually sampled initial and boundary conditions. (b-f) The resulting spline after 3,9,12 and 14 gradient descent iterations correspondingly.}
	\label{fig:StepbyStepIteration_waw}
\end{figure}

Figure \ref{fig:CostFunctionPlot}d demonstrates the values of the cost function in equation \eqref{eq:final_functional} through the iterations performed for a single stroke restoration. Figures \ref{fig:CostFunctionPlot}a-c illustrates the energy terms, mentioned above, after factorization.

\begin{figure}[ht]
	\begin{centering}
		\includegraphics[width={\linewidth}]{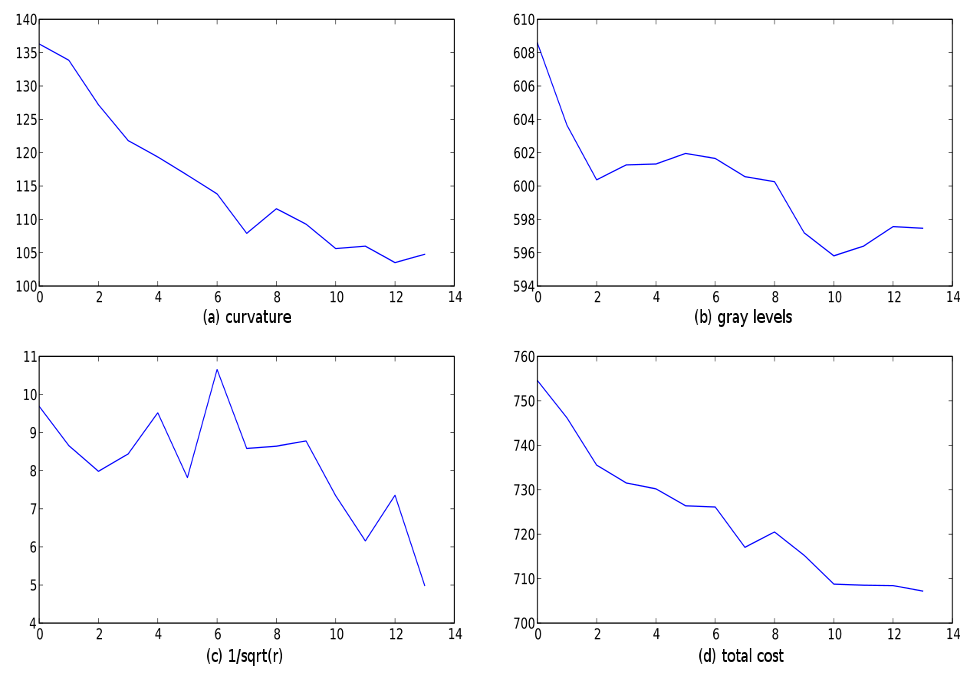}
		\par\end{centering}
	\caption{Development of the cost function terms of equation \eqref{eq:final_functional}. (a) the term $F_K$ which is the integral over the curvature (i.e., internal energy)  (b) the term $F_E^{(S)}$ which is the integration of the gray level under the curve. (c) the term $F_E^{(r)}$ assuring the growth of the radii. (d) the cost functional.}
	\label{fig:CostFunctionPlot}
\end{figure}

\section{Experimental Results}
\label{sec:Experimental}
The \textit{stroke restoration} algorithm, was designed in order to overcome the difficulties of the very noisy ancient-documents media. The basic motivation for this research comes from First Temple period Hebrew inscriptions which are written in ink over clay shards (i.e, ostraca). We have, therefore, put our method to a test using various characters taken from several ostraca as input data. These inscriptions were excavated in Tel Arad during the 1960's by Yohanan Aharoni \cite{AharoniArad} and are dated to the beginning of the $6^{th}$ century BCE (see for example \cite{AharoniArad,na2011elyashib}).

In all of the following experiments we used the same parameters. This was successful as we first performed histogram stretching and only then applied Algorithm \ref{alg:StrokeRestoration} to the enhanced image (see Figure \ref{fig:AlepHoriz}c). Thus, the differences of brightness and contrast between various images do not affect the scaling of our cost function. This step was not necessary but it saved the time of balancing our parameters for each and every image. The parameters used for a $\sim 350 \times 350$ pixels character were: \textit{maximal radius} = 50 ; \textit{minimal radius} = 3 ; the cost function parameters from equation \eqref{eq:final_functional} $c_1=2 , c_2=2000 , c_3=50$  (they were set by trial and error). The number of iterations performed was 14.

Examples of reconstructed strokes are presented in Figures \ref{fig:AlepHoriz} and \ref{fig:AlepVert}. In both cases the strokes belong to the same "alep" character from Arad Ostracon No.1. In the example presented in Figure \ref{fig:AlepHoriz} we can see that where two written lines overlap there is a need to sample delicately in order to preserve the separation between the two lines. Noteworthy is the authenticity of both restorations.

\begin{figure}[ht]
	\begin{centering}
		\includegraphics[width={\linewidth}]{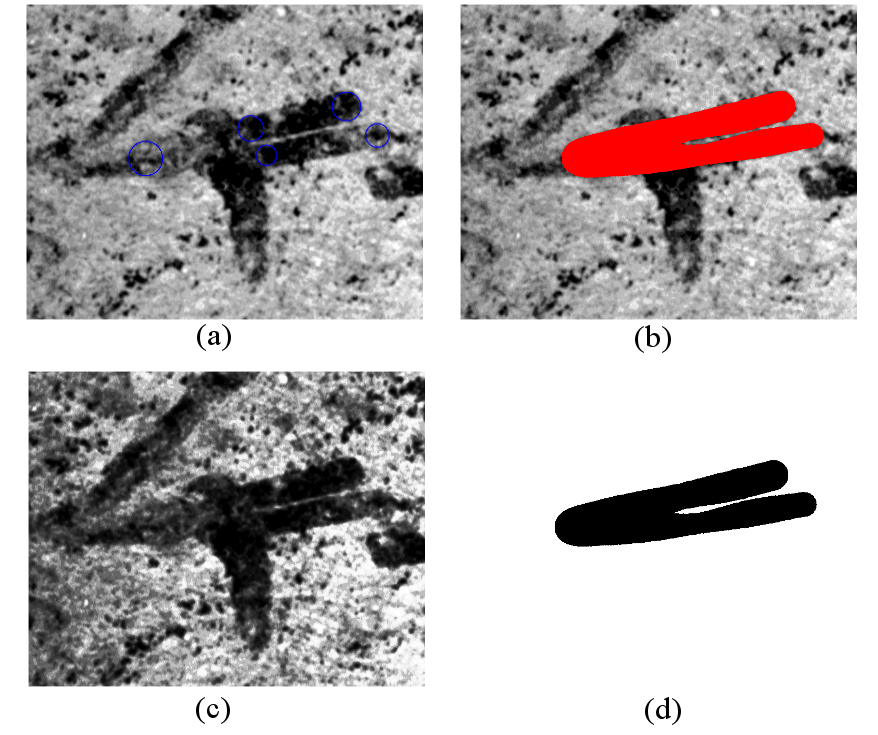}
		\par\end{centering}
	\caption{First stroke of an "alep" character taken from Arad Ostracon No.1. (a) the initial and boundary points. (b) the resulting spline after performing 14 steps of the gradient descent. (c) the character's image after histogram stretching. (d) the resulting binariztion.}
	\label{fig:AlepHoriz}
\end{figure}

\begin{figure}[ht]
	\begin{centering}
		\includegraphics[width={\linewidth}]{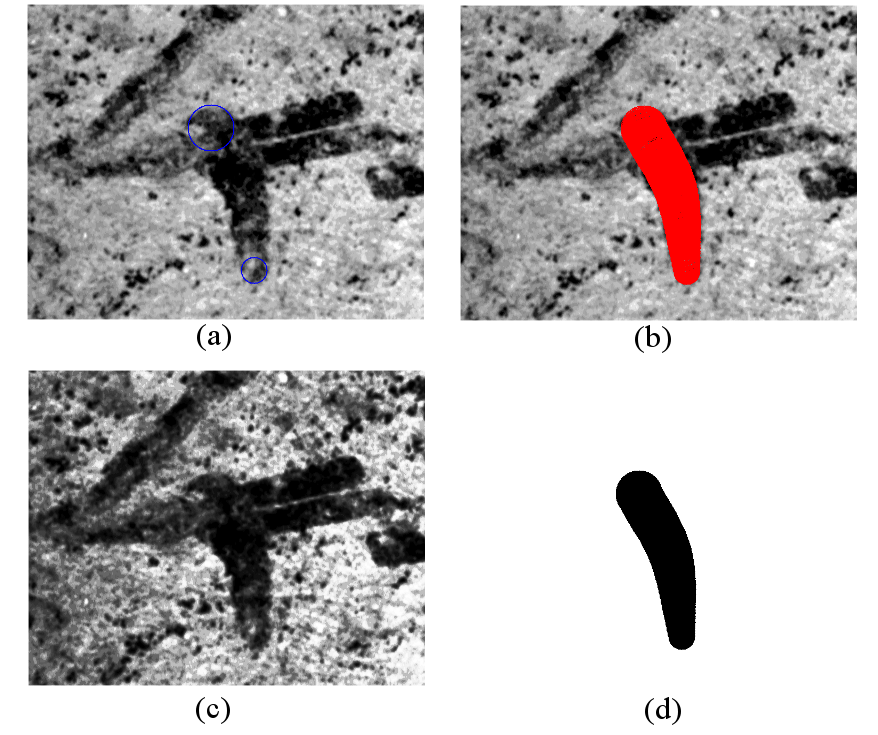}
		\par\end{centering}
	\caption{Second stroke of an "alep" character taken from Arad Ostracon No.1. (a) the initial and boundary points. (b) the resulting spline after performing 14 steps of the gradient descent. (c) the character's image after histogram stretching. (d) the resulting binariztion.}
	\label{fig:AlepVert}
\end{figure}

By overlaying the two reconstructed strokes one over the other we get a very clean binarization (Figure \ref{fig:AlepCombined}c). A comparison between the fully reconstructed "alep" character to the handmade facsimile is presented in Figure \ref{fig:AlepCombined}. It is apparent that the restoration gives a favorable result. Moreover, the reconstructed stroke has an analytic representation which enables us to deduce accurate mathematical properties of the curve, such as the curvature.

\begin{figure}[ht]
	\begin{centering}
		\includegraphics[width={\linewidth}]{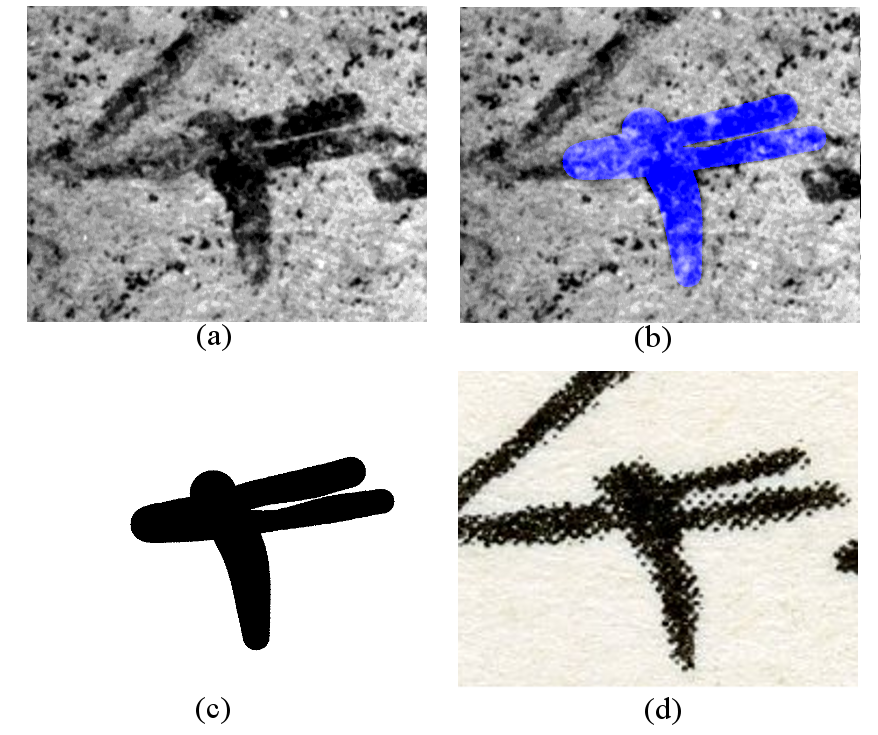}
		\par\end{centering}
	\caption{Restoration of a complete “alep” character (a) the original image (b) the restoration overlaid upon the image (c) the two reconstructed binary strokes overlaid (d) a handmade facsimile drawn by Yohanan Aharoni \cite{AharoniArad}.}
	\label{fig:AlepCombined}
\end{figure}

Another reconstructed character from the same ostracon is presented in the same fashion in figures \ref{fig:ShinLeft}-\ref{fig:ShinCombined}. This time the reconstructed character is "shin". For other examples of reconstructed "waw" characters see \ref{fig:Waws}

\begin{figure}[ht]
	\begin{centering}
		\includegraphics[width={\linewidth}]{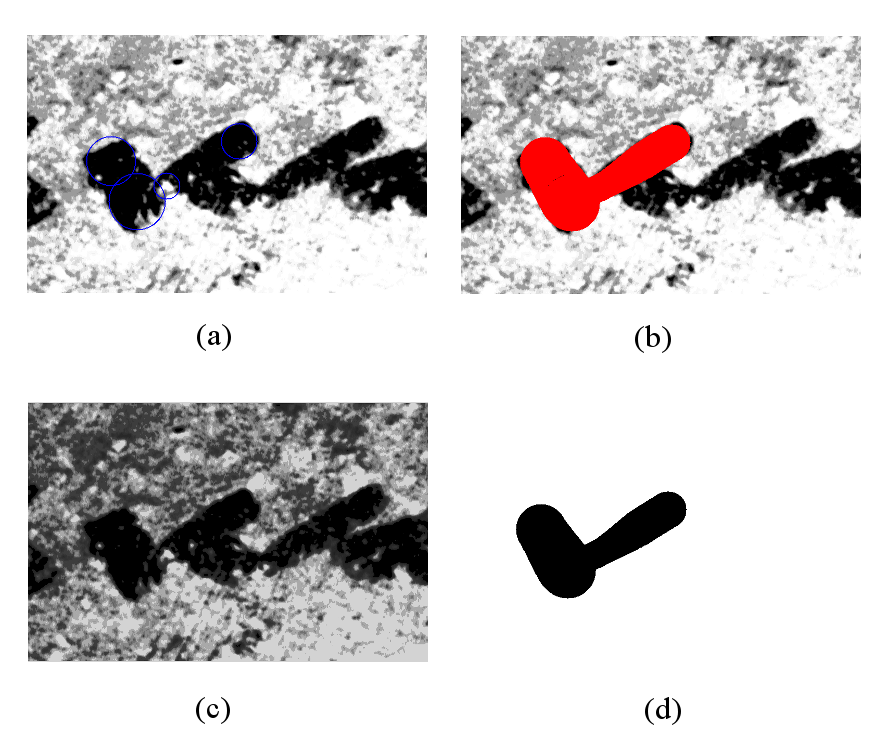}
		\par\end{centering}
	\caption{First stroke of a "shin" character taken from Arad Ostracon No.1. (a) the initial and boundary points. (b) the resulting spline after performing 14 steps of the gradient descent. (c) the character's image after histogram stretching. (d) the resulting binariztion.}
	\label{fig:ShinLeft}
\end{figure}

\begin{figure}[ht]
	\begin{centering}
		\includegraphics[width={\linewidth}]{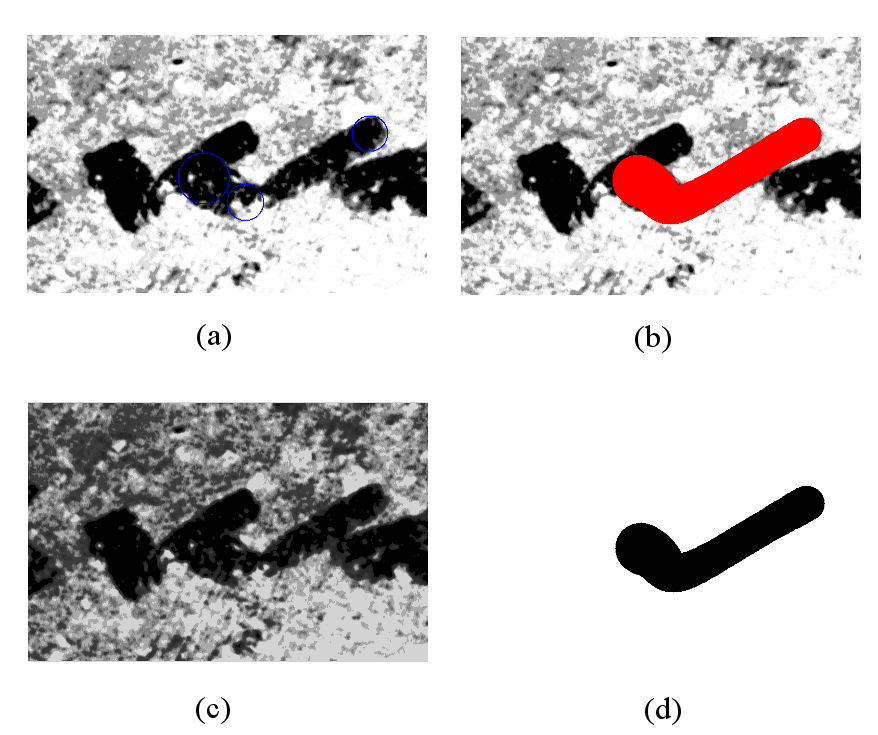}
		\par\end{centering}
	\caption{Second stroke of a "shin" character taken from Arad Ostracon No.1. (a) the initial and boundary points. (b) the resulting spline after performing 14 steps of the gradient descent. (c) the character's image after histogram stretching. (d) the resulting binariztion.}
	\label{fig:ShinRight}
\end{figure}

\begin{figure}[ht]
	\begin{centering}
		\includegraphics[width={\linewidth}]{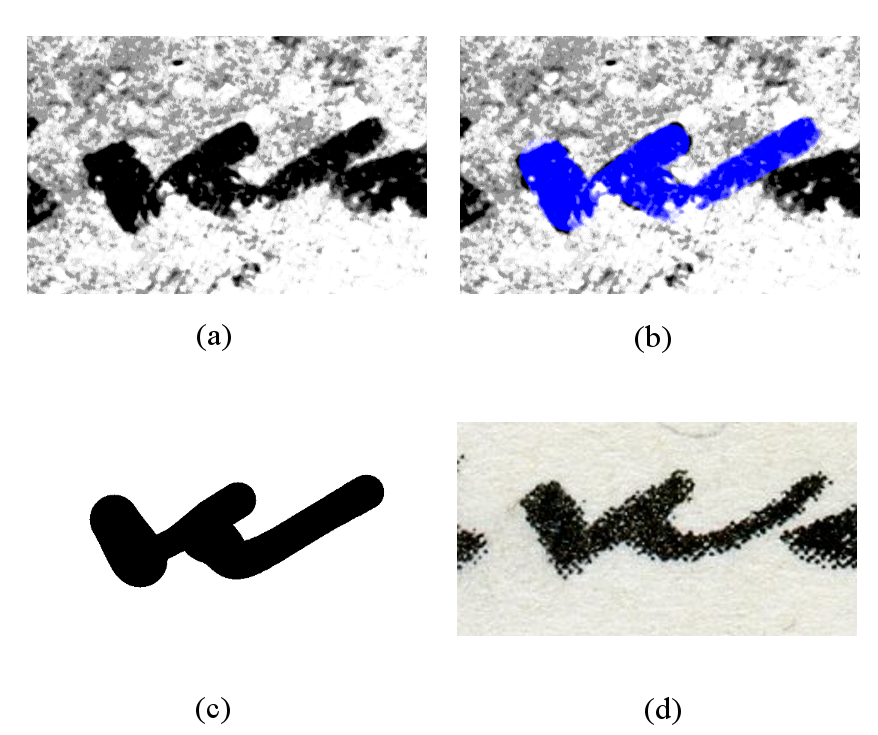}
		\par\end{centering}
	\caption{Restoration of a complete "shin" character (a) the original image (b) the restoration overlaid upon the image (c) the two reconstructed binary strokes overlaid (d) a handmade facsimile drawn by Yohanan Aharoni \cite{AharoniArad}.}
	\label{fig:ShinCombined}
\end{figure}

\begin{figure}[ht]
	\begin{centering}
		\includegraphics[width={\linewidth}]{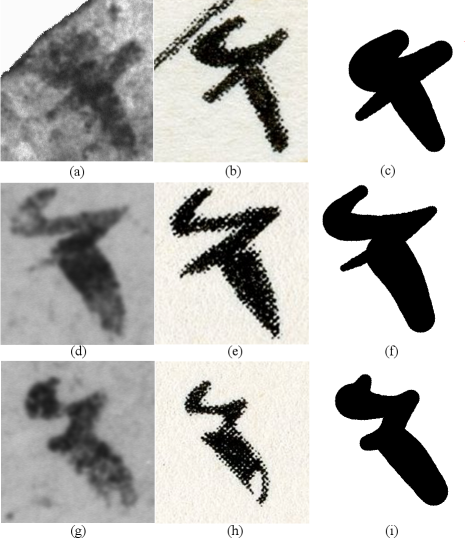}
		\par\end{centering}
	\caption{Restoration of three "waw" characters using Algorithm 1 the first line from Arad Ostracon No. 1 and the second and third lines are from Arad Ostracon No. 2: (a,d,g) are the original images of the characters; (b,e,h) are the handmade facsimiles; (c,f,i) are the restorations of Algorithm \eqref{alg:StrokeRestoration}.}
	\label{fig:Waws}
\end{figure}

Furthermore, the stroke restoration algorithm was utilized by epigraphers to produce a new reconstruction of an entire inscription. The inscription, which was found in the excavations of the City of David at Jerusalem (i.e., Ophel) is presented in Fig. \ref{fig:Opehl}. 

\begin{figure}[ht]
	\begin{centering}
		\includegraphics[width={\linewidth}]{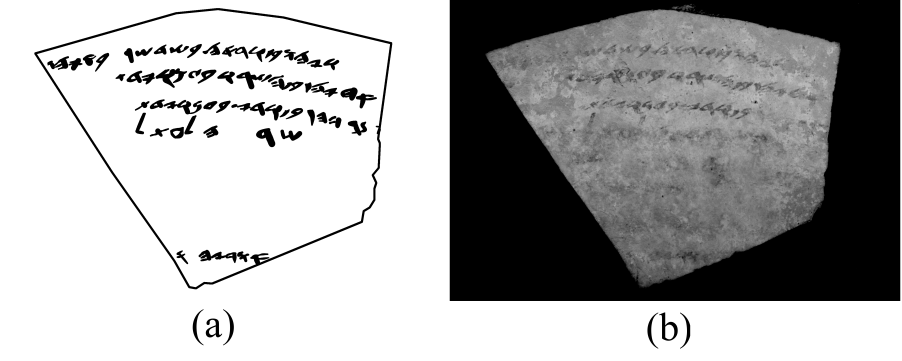}
		\par\end{centering}
	\caption{The Ophel (Jerusalem) ostracon: (a) a facsimile created by utilizing the stroke restoration algorithm; (b) the original grayscale image. The restoration, which was taken from \cite{ms2015ophel}, had been produced by Mrs. Shira Faigenbaum-Golovin and is presented here with her courtesy.}
	\label{fig:Opehl}
\end{figure}

\section{Concluding Remarks}

Since handmade restorations of characters are pivotal to the understanding of ancient Hebrew writing there arises a need to aid the epigraphers in producing a more standardized restorations and reduce the inherent bias in the process. The research presented above describes a new computer aided approach to segmentation and restoration of incomplete character strokes in a noisy background. The soundness of the restoration method was tested in various cases and produced clean and reasonable results. These restorations yield analytic representations of the strokes that can be utilized to compare given strokes to one another.

It should be noted, that the \textit{stroke restoration} algorithm was tested on more than \textbf{1000} different characters. Approximately 500 characters were reconstructed from Tel Arad inscriptions, 500 from Samaria inscriptions and the rest are restorations of the Ophel (Jerusalem) ostracon mentioned above. Upon manual inspection of the results about 11.5\% were precluded as they did not adhere to the original images sufficiently. The restorations from Tel Arad were used in a computerized paleographic investigation regarding the literacy rates in the Kingdom of Judah ca. 600 BCE \cite{AradPNAS2016}. In addition, the Ophel ostracon restorations were used to present a new reading of the inscription \cite{ms2015ophel}, while the rest are intended for a future investigations of the dissemination of writing in the Kingdom of Israel in the $8^{th}$ century BCE.

Although the algorithm presented in this work was developed to tackle difficulties stemming from Historical Document Analysis it may be useful for other fields of research. The task of the epigrapher in many ways resembles the one of the forensics expert, trying to decide whether a specific document was written by some suspect. We therefore, assume that the approach presented above could be applicable in fields related to computerized forensics.

The problem of identifying a writer of a specific document is still open. There is still a lot to be done as the identification rates of the state-of-the-art algorithms are still very low \cite{louloudis2012icfhr}. Using the analytic representations of characters could significantly enlarge the amount of features extracted. The obvious obstacle of applying our methodology to that field is the fact that it involves manual sampling of the characters. In accordance to that, in order to use similar techniques for the case of writer identification there arises a need to develop an automatic sampling procedure.

Another open subject is offline signature verification and identification of forgeries. Forgery detection is interesting not only from the forensic science perspective but for the historical purposes as well, since the market of forged antiquities flourish (specifically when referring to First Temple Period in Israel – for example see \cite{rollston2005navigating}).

\section*{Acknowledgement}
The research leading to the results reported here received funding from the Israel Science Foundation F.I.R.S.T. (Bikura) Individual Grant no. 644/08 as well as the Israel Science Foundation Grant no. 1457/13. The research was also partially funded by the European Research Council under the European Community’s Seventh Framework Programme (FP7/2007-2013)/
ERC grant agreement no. 229418, and by an Early Israel grant (New Horizons project), Tel Aviv University. This study was also supported by a generous donation of Mr. Jacques Chahine, made through the French Friends of Tel Aviv University. 

The ostraca images as well as the handmade facsimiles used in the article are courtesy of the Institute of Archaeology, Tel Aviv University; and of the Israel Antiquities Authority. 

We thank our colleagues at Tel Aviv University: Israel Finkelstein, Eliezer Piasetzky for presenting us with the challenges of ancient Hebrew inscriptions and for their insights; the priceless advice of Arie Shaus and Shira Faigenbaum-Golovin are appreciated dearly.

\section*{Appendix A - Detailed Analysis of The Fidelity Terms}
In this appendix we wish to take a closer look at the first and second terms of the right hand side of equation \eqref{eq:final_functional}:
 \begin{equation}
	\begin{array}{c}
		\lim_{c_3 \rightarrow 0} F = F_{E} = c_{1}\intop_{0}^{T}\frac{S_{I}(t)}{(r(t))^{2}}dt +  c_{2}\intop_{0}^{T}\frac{1}{\sqrt{r(t)}}dt \\
		F_{E} = F_{E}^{(S)} + F_{E}^{(r)}
	\end{array}
	.\end{equation} 
In order to give a full explanation to the power of $r(t)$ in $F_{E}^{(r)}$ we wish to rewrite it as
 \[
F_{E}^{(r)} = c_{2}\intop_{0}^{T}\frac{1}{r(t)^{\alpha}}dt
,\] 
and try to find an optimal power $0<\alpha$ that will suite our case.

\textbf{Assumptions:} For the clarity of our analysis we suppose that the handwriting is colored black (i.e, pixel value 0) over white background (i.e, pixel value 1). This is a minor assumption since all images can be binarized. Moreover, it is possible to perform a similar analysis using the fact that the stroke color is relatively dark while the background color is bright.

Accordingly, as long as the curve remains completely within the stroke
boundaries the first term $F_{E}^{(S)}$ is identically zero. If the
radius $r(t)$ at some point $(x(t),y(t))^T$ exceeds the radius of
writing (denoted $R$) then $F_{E}^{(S)}$ grows. Had we $F_{E}^{(S)}$
alone (by setting $c_{2}=0$) then the optimum could be reached at
any radius $r(t)$ that keeps the curve completely within the stroke
boundaries. Therefore, $F_{E}^{(r)}=\intop_{0}^{T}\frac{1}{r^{\alpha}}dt$
is necessary to ensure that the radius wishes to grow to the boundary
of the stroke, since $F_{E}^{(r)}$ decreases when $r(t)$ grows.

For further simplification we require that the discs' centers are
situated along the medial axis of the stroke image (i.e., skeleton). Since we force our
restoration to comply with some initial conditions
(e.g., beginning and ending of the stroke) which prevents the solution
from collapsing to a null solution, locating the discs' centers outside
the stroke image cannot result in the globally minimal curve. Therefore,
if the centers are within the stroke their most favorable position
would be at the true centers (i.e., the skeleton) so that $F_{E}^{(r)}$
will be minimal and $F_{E}^{(S)}$ will remain zero. The term $F_{E}^{(r)}$
forces $r(t)$ to grow to the limits of the writing radius, denoted by $R$. Moreover,
$F_{E}^{(r)}$ prevents $r(t)$ from collapsing to zero. 

Our final demand is that, aside from the initial conditions,
the stroke will have a low curvature, which means that locally it
resembles a straight line. This condition is not fulfilled at points
where the stroke curvature is extreme. For this reason we require
these points to be part of our initial conditions. 

To summarize our assumptions for the analysis of $F_{E}^{(S)}$ and $F_{E}^{(r)}$: 
\begin{enumerate}
	\item The stroke is colored black (pixel value 0) while the background white
	(pixel value 1).
	\item The discs' centers are situated at the skeleton of the stroke image.
	\item The original stroke\textquoteright{}s curvature will be low and will
	therefore locally resemble a straight line.
	\item We require that the stroke comply with some initial conditions
	(e.g., beginning and ending of stroke and curvature extremal points)
	to prevent null solutions as an optional result (a more detailed discussion
	regarding the initial conditions is presented in Section
	3.3)
\end{enumerate}
\textbf{Rewriting $G_{I}(t)$:} earlier we stated that $G_{I}(t)=\sum_{(p,q)\in S(t)} I(p,q) \approx \intop_{(p,q)\in S(t)}I(p,q)dt$, that is $G_{I}(t)$ is the summation of all gray level values inside
the disc. Using assumptions 1 and 2 we can deduce that as long as
$r<R$ we have $G_{I}(t)=0$ and by using assumption 3 and straightforward
calculations for $r\geq R$ we get 
 \[
G_{I}(t)\approx (\theta-\sin(\theta))\cdot r^{2} 
,\] 
where $\theta$ denotes the segment's angle which in our case equals
$\theta=2\arccos(\frac{R}{r})$ (see Figure \ref{fig:sector}). Hence, $\frac{G_{I}(t)}{r^{2}}$ is a normalized deviation measure.

\begin{figure}
	\begin{centering}
		\includegraphics[width={0.6\linewidth}]{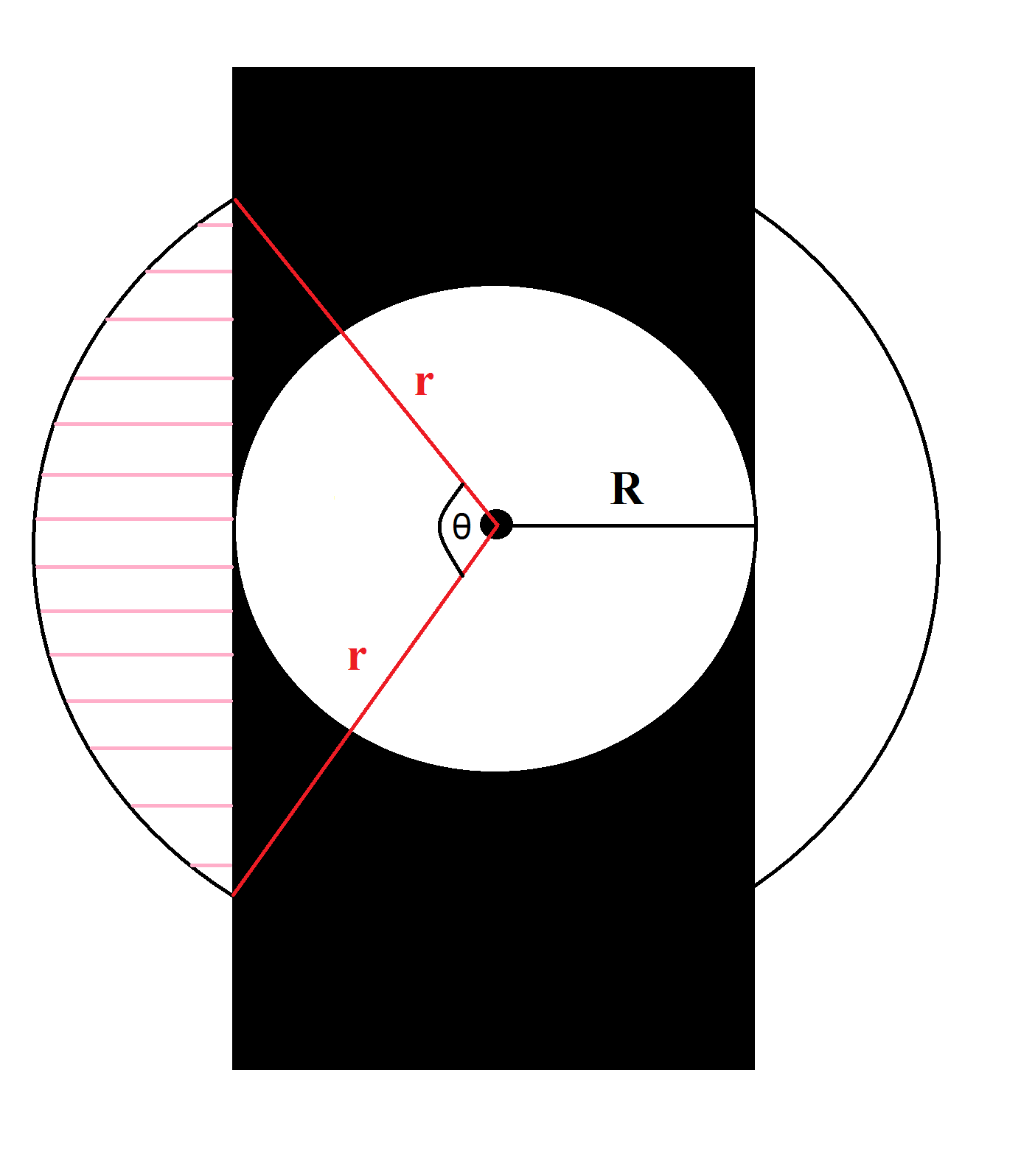}
		\par\end{centering}
	
	\caption{The black area illustrates a straight stroke. The inner disc with
		radius $R$ is the writing disc, whereas the larger disc with radius $R$ is
		the disc $S(t)$ for some specific $t$. Marked with pink lines is the segment exceeding the stroke
		to the left. The right segment has the same area.\label{fig:sector}}
\end{figure}
\textbf{Choosing the parameter $\alpha$:} As our aim is that the functional will be minimized as close as possible to $r=R$, we wish the discs to be as large as possible within the limits of the foreground. Accordingly, it is clear that $\alpha$ should be larger than zero (i.e.,  $\alpha>0$). However, in order to decide what is the most fitting value we look at a single disc $D(x_0, y_0, r)$ for some constant $x_0,y_0$ and a varying radius $r$ where the stroke is straight and the writing radius $R$ is constant(see Figure \ref{fig:sector} for an illustration). For this case we can write $F_E$ as:
 \begin{equation}
	F_E = \frac{c_2}{r^{\alpha}} + c_1 \cdot \frac{G_{I}}{r^2} \approx \frac{c_2}{r^{\alpha}} + c_1 \cdot
	\begin{cases} 
		0, & r < R  \\
		\theta  - sin(\theta), & r \geq R
	\end{cases}
	\label{eq:IeforAlpha}
	.\end{equation} 
Notice that $G_{I} \geq 0$ and when $r=R$ we have that $\theta = 0$. Therefore, $G_{I}$ is continuous in $r$ and as we will show it is also differentiable at $r=R$. Since 
\[F_E(r \leq R) = \frac{c_2}{r^{\alpha}},\]
local extremal points can appear only in $(R, \infty)$ where substituting $\theta$ with $\arccos (\frac{R}{r})$ into equation \eqref{eq:IeforAlpha} yields:
 \begin{equation}
	F_E = c_1 \cdot \left( 2 \arccos \frac{R}{r} - sin(2 \arccos \frac{R}{r}) \right) + \frac{c_2}{r^\alpha}
	.\end{equation} 
Differentiating this expression with respect to $r$ gives:
 \begin{equation}
	\frac{\partial F_E}{\partial r} = \frac{4 c_1 R \sqrt{1 - \frac{R^2}{r^2}}}{r^2} - c_2 \alpha r^{-\alpha - 1}
	\label{eq:IeDeriv}
	,\end{equation} 
and we can find the extremum for each predetermined $\alpha > 0$.
For example let us assume that $\alpha = 1$ then instead of \eqref{eq:IeDeriv} we get:
 \begin{equation}
	\frac{\partial F_E}{\partial r} = \frac{4 c_1 R \sqrt{1 - \frac{R^2}{r^2}} - c_2}{r^2} 
	.\end{equation} 
Hence, for this case, the minimum is found at:
 \begin{equation}
	r = \frac{4c_1R^2}{\sqrt{16 R^2 c_1^2 - c_2^2}}
	,\end{equation} 
which is not far from our desired minimum at $r=R$.
Taking a closer look at \eqref{eq:IeDeriv} we can see that $\frac{\partial F_E}{\partial r} (r = R^+) = \frac{\partial F_E}{\partial r} (r = R^-) = -c_2 \alpha R^{- \alpha - 1} < 0$. Therefore, the faster $-c_2 \alpha R^{- \alpha - 1}$ decreases the root of equation \eqref{eq:IeDeriv} will be closer to $R$.
Numerical experiments performed with $F_E$ indicated that as long as $\alpha \geq 0.5$ the choice of $\alpha$'s value does not affect the the results significantly.

Since our initial motivation is the restoration of a stroke in a highly noisy environment, where the ink is sometimes partly erased, we chose a conservative parameter of $\alpha = 0.5$ ,  which has proven to be more resistant to noise than larger values. 
\bibliographystyle{elsarticle-num}
\bibliography{mybib}
\end{document}